\documentclass[letterpaper, 10 pt, conference]{ieeeconf}  
\IEEEoverridecommandlockouts                              
\title{\LARGE \bf On Improved Statistical Accuracy of Low-Order Polynomial Chaos Approximations \thanks{\textbf{Funding:} This work was funded by the National Science Foundation (NSF) under contract no.~1762825.}}

\author{Vedang M. Deshpande$^{1}$ and Raktim Bhattacharya$^{2}$
\thanks{$^{1}$Research/Technical Staff, Mitsubishi Electric Research Laboratories, Cambridge, MA, USA. This work was completed during doctoral studies at Texas~A\&M~University. Email: \texttt{deshpande@merl.com}}
\thanks{$^{2}$Professor, Aerospace Engineering, Texas A\&M University. Email: \texttt{raktim@tamu.edu}.}
}
\usepackage{lipsum}
\usepackage{amsfonts}
\usepackage{graphicx}
\usepackage{epstopdf}
\usepackage{algorithm}
\usepackage{algpseudocode}
\usepackage{fullpage}
\usepackage[none]{hyphenat}

\ifpdf
  \DeclareGraphicsExtensions{.eps,.pdf,.png,.jpg}
\else
  \DeclareGraphicsExtensions{.eps}
\fi

\usepackage{amsopn}
\DeclareMathOperator{\diag}{diag}

\usepackage{subcaption,amsthm,amssymb}
\usepackage{amsmath}

\newtheorem{theorem}{Theorem}

\newtheorem{remark}[theorem]{Remark}

\newcommand{\real}{\mathbb{R}}
\newcommand{\vo}[1]{\boldsymbol{#1}}
\newcommand{\mo}[1]{\boldsymbol{#1}}

\newcommand{\x}{\vo{x}}
\newcommand{\F}{\vo{F}}
\newcommand{\EE}{\vo{E}}
\renewcommand{\L}{\vo{L}}

\newcommand{\param}{\vo{\Delta}}
\newcommand{\fhat}{\vo{\hat{f}}}
\newcommand{\fhatp}{\vo{\hat{f}}(\param)}
\newcommand{\fhatpt}{\vo{\hat{f}}^T(\param)}
\newcommand{\fp}{\vo{f}(\param)}
\newcommand{\fpt}{\vo{f}^T(\param)}

\newcommand{\phipt}{\vo{\Phi}^T(\param)}

\newcommand{\xdot}{\dot{\vo{x}}}

\newcommand{\pdfp}{p(\param)}
\newcommand{\set}[1]{\mathcal{#1}}
\newcommand{\Exp}[1]{\mathbb{E}\left[#1\right]}
\newcommand{\E}[1]{\Exp{#1}}

\newcommand{\basis}[2]{\phi_{#1}\left(#2\right)}
\newcommand{\I}[1]{\vo{I}_{#1}}
\newcommand{\X}{\vo{X}}

\newcommand{\W}{\vo{W}}

\newcommand{\R}{\vo{R}}
\newcommand{\Q}{\vo{Q}}

\newcommand{\Phin}[1]{\vo{\Psi}(\param)}
\newcommand{\Phint}[1]{\vo{\Psi}^T(\param)}

\newcommand{\Real}{\mathbb{R}}

\renewcommand{\vec}[1]{\textbf{vec}\left(#1\right)}
\newcommand{\xpc}{\x_{\text{pc}}}
\newcommand{\xgpc}{\x_{\text{gpc}}}
\newcommand{\xcpc}{\x_{\text{cpc}}}

\newcommand{\inner}[1]{\left\langle \vo{e}\phi_i\right\rangle}

\newcommand{\eqnlabel}[1]{\label{eqn:#1}}
\newcommand{\figlabel}[1]{\label{fig:#1}}

\newcommand{\eqn}[1]{(\ref{eqn:#1})}
\newcommand{\Eqn}[1]{(\ref{eqn:#1})}
\newcommand{\fig}[1]{Fig. \ref{fig:#1}}
\newcommand{\Fig}[1]{Fig. \ref{fig:#1}}

\DeclareMathAlphabet{\mathbfsf}{\encodingdefault}{\sfdefault}{bx}{n}

\renewcommand{\H}{\vo{H}}

\newcommand{\domain}[1]{\set{D}}
\newcommand{\domainD}{\set{D}_{\param}}

\newcommand{\trace}[1]{\textbf{tr}\left( #1 \right)}
\newcommand{\tr}[1]{\textbf{tr}\:#1}

\newcommand{\algo}[1]{Algorithm \ref{#1}}
\newcommand{\lfp}[1]{linear-fit-predict }

\begin{document}
\maketitle
\thispagestyle{empty}
\pagestyle{empty}

\begin{abstract}
Polynomial chaos expansions provide surrogate models for stochastic systems, with coefficients typically derived using Galerkin projection, stochastic collocation, or least squares approximation. These traditional approaches often fail to accurately capture statistical moments without resorting to high-order approximations. We propose a constrained optimization framework that modifies standard techniques to determine polynomial chaos coefficients that precisely recover the first two statistical moments. The effectiveness of our approach is demonstrated on several candidate algebraic functions of random variables, showing significant improvements in statistical accuracy even with low-order approximations.
\end{abstract}
\begin{keywords}
 polynomial chaos, Galerkin projection, uncertainty quantification
\end{keywords}
\section{Introduction}\label{sec:intro}
Polynomial chaos (PC) expansions provide a systematic framework for representing and propagating uncertainty in dynamical systems. By expressing random variables as truncated series of orthogonal polynomials with respect to known probability measures~\cite{wiener_askey}, PC enables tractable analysis of stochastic dynamics through deterministic surrogates. In control applications, where robustness and performance often depend critically on model uncertainties, PC-based methods have been widely employed to model, propagate, and quantify uncertainty~\cite{Shen2020PCReview,kim2021RobustKF,hover2006PCStability,Bhattacharya2019RobustLQR}. These surrogate models allow efficient analysis of complex and computationally intensive systems, reducing reliance on costly Monte Carlo simulations. 

Arbitrary probability distributions can be represented through expansions
in easy-to-sample distributions (e.g. uniform or Gaussian). For stochastic differential equations, Galerkin projection (GP) yields deterministic evolution equations for the expansion coefficients, solvable with standard numerical methods, and the effectiveness of generalized PC (gPC) in such problems is well established~\cite{dXiu_jcp,dXiu_bc}.
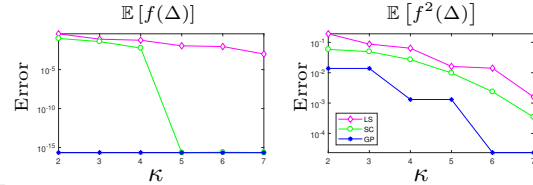
\begin{figure}[tb!]\centering
\includegraphics[width=0.4\textwidth]{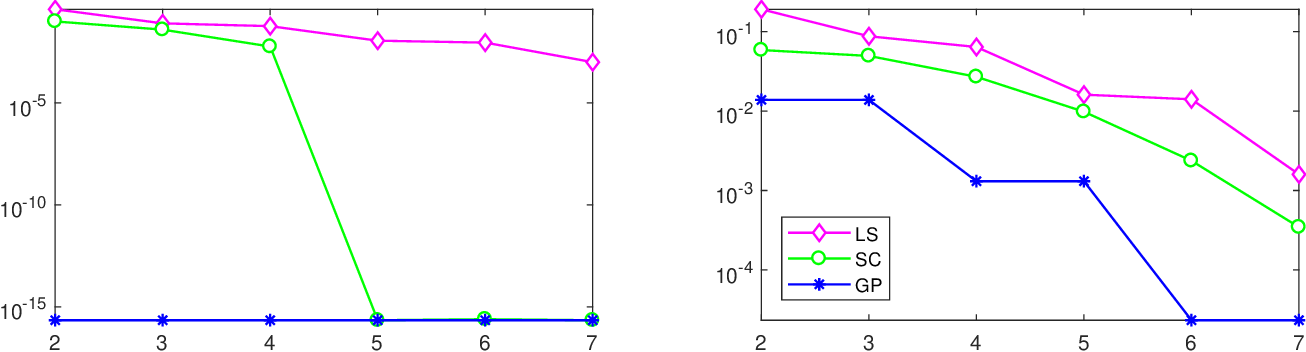}
\begin{picture}(0,0)
        \put(-200,20){\rotatebox{90} {{\scriptsize Error}}} 
        \put(-95,20){\rotatebox{90} {{\scriptsize Error}}} 
        \put(-150,-6){{$\kappa$}}
        \put(-50,-6){{$\kappa$}}
        \put(-160,55){{\scriptsize $\E{f(\Delta)}$}}
        \put(-60,55){{\scriptsize $\E{f^2(\Delta)}$}}
\end{picture}
\caption{\small Absolute errors in estimated moments of $f(\Delta)=\Delta^8$ using LS, SC, and GP formulations for different approximation orders ($\kappa$), where $\Delta$ is uniformly distributed over $[-1,1]$.}
\figlabel{pcErrors}
\vspace{-20pt}
\end{figure}
GP requires inner-product integrals, which may be evaluated intrusively by modifying the governing equations or non-intrusively using sampling-based quadrature. Non-intrusive approaches are particularly attractive in practice, as they treat the simulator as a black box. Among these, stochastic collocation (SC) and least-squares (LS) regression are widely used, with comparative studies provided in~\cite{constantine,sc2009eldred}. While theoretically an infinite-order PC expansion converges to any square-integrable stochastic process~\cite{wiener_askey,cameron_martin}, practical implementations require truncation to finite order. This truncation is a major source of error, often leading to inaccurate estimates of statistical moments. In particular, as illustrated in \fig{pcErrors}, (i) GP recovers the first moment exactly but may exhibit significant errors in higher moments, especially at low orders, while (ii) SC and LS generally fail to recover even the first moment accurately, with LS being least accurate.

Since many control applications often rely primarily on the first two moments (e.g., mean and variance) for analysis and design, there is strong motivation to develop surrogate modeling techniques that guarantee their accuracy at low polynomial orders. 
To this end, we propose a framework for computing PC coefficients that enforces exact matching of the first two moments with those of the underlying random variables. 
This framework yields constrained variants of the GP and LS formulations, namely, the \textit{constrained} $\mathcal{L}_2$ and \textit{constrained} $l_2$ methods, respectively. In contrast, SC admits no such modification, as its interpolation-based construction leaves no degrees of freedom for imposing additional constraints. 

The remainder of the paper is organized as follows. Sections \ref{sec:conL2} and \ref{sec:conl2} present the constrained $\mathcal{L}_2$ and constrained $l_2$ formulations, stated in Theorems \ref{thm:gp1}, \ref{thm:bestL2}, and \ref{thm:bestl2}, together with numerical results. Concluding remarks are provided in Section \ref{sec:conclusions}.

\section{Constrained $\mathcal{L}_2$-optimal Approximation}\label{sec:conL2}
\begin{figure*}[tb!]
    \centering
    \includegraphics[width=0.9\textwidth]{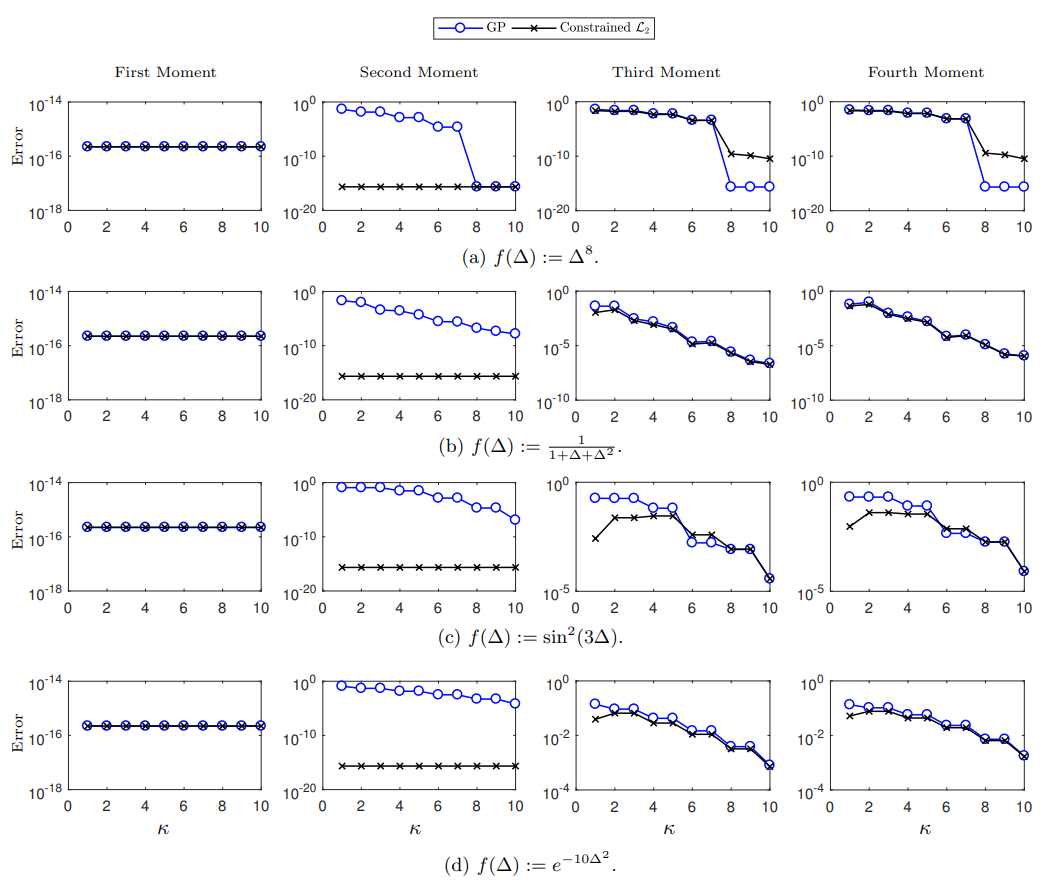}
    \caption{Accuracy of PC approximations of different orders ($\kappa$) for various functions using standard GP and constrained $\mathcal{L}_2$ formulations. Plots show absolute error, $\left|\mathbb{E}\big[f^m(\param)\big] - \mathbb{E}\big[\hat{f}^m(\param)\big]\right|$, for $m=$ 1, 2, 3, and 4, where $\hat{f}$ is the approximated function, and $\Delta$ is uniformly distributed over $[-1,1]$.}
    \label{fig:conGPC}
\end{figure*}

Let $\param\in\domainD\subseteq\real^d$ be a random vector with probability density function $\pdfp$. Let $\vo{f}(\param):\domainD\mapsto\real^n$ be a vector function whose components are square integrable, i.e.
$$
\int_{\domainD} \vo{f}^2(\param)\pdfp d\param < \infty,
$$
where the integral and inequality are evaluated elementwise. 
The objective here is to approximate $\vo{f}(\param)$ using PC expansions, i.e. find PC coefficients $\vo{f}_i\in\real^n$ such that
\begin{align}
\vo{f}(\param) \approx \fhat(\param) :=  \sum_{i=0}^N \vo{f}_i\phi_i(\param), \eqnlabel{pcF}
\end{align}
for a known set of basis polynomials $\{\phi_i(\param)\}$, which is selected based on the type of random distribution \cite{wiener_askey}.
The highest degree of polynomials $\{\phi_i(\param)\}$ is called the \textit{order of approximation}, and is denoted by $\kappa$. The number of basis functions ($N+1$) is related to the dimension ($d$) of the random vector $\param$ and the order of approximation $\kappa$ by $(N+1)=\frac{(d+\kappa)!}{d! \: \kappa !}$. We define $\mo{\Phi}(\param)$ as
\begin{align}
\mo{\Phi}(\param) &:= \begin{bmatrix}\basis{0}{\param}, & \cdots, & \basis{N}{\param}\end{bmatrix}^T.
\eqnlabel{Phi:def}
\end{align}
We also define matrix $\F\in\real^{n\times(N+1)}$, with polynomial chaos coefficients $\vo{f}_i$, as its columns, $\F := \begin{bmatrix} \vo{f}_0, & \cdots, & \vo{f}_N \end{bmatrix}$.
Therefore, $\fhat(\param)$ can be compactly written as
\begin{align}
\fhat(\param) = \F\mo{\Phi}(\param) \eqnlabel{compactFpc}.
\end{align}
The coefficients $\vo{f}_i$ are computed via GP, LS, or SC. We next characterize the moment errors associated with these methods and present new formulations that achieve improved statistical accuracy.

In the GP formulation, or the $\mathcal{L}_2$ formulation, the basis functions $\phi_i(\param)$ are $d$-variate polynomials that are orthogonal with respect to $\pdfp$, i.e.,
\begin{multline}
\E{\phi_i(\param)\phi_j(\param)} := \int_{\domainD} \phi_i(\param)\phi_j(\param) \pdfp d\param= 0, \\\text{ for } i \neq j.
\eqnlabel{orthoBasis}
\end{multline}
Since $\phi_0(\param) = 1$ for orthogonal basis, it follows that $\E{\vo{\Phi}(\param)} = [1,0,\cdots,0]^T$.

Optimal coefficients $\vo{f}_i$ are determined by projecting the error $\vo{e}(\param):=\vo{f}(\param) - \fhat(\param)$ against each basis polynomial and setting it to zero,
$
\E{\vo{e}(\param) \phi_i(\param)} = \vo{0}, \text{ for } i=0,\cdots,N,
$
or more compactly
\begin{align}
\E{\vo{f}(\param)\vo{\Phi}^T(\param)} - \F\E{\mo{\Phi}(\param)\vo{\Phi}^T(\param)} = \vo{0}.
\eqnlabel{gp:errProj}
\end{align}
From orthogonality \eqn{orthoBasis},
\begin{multline}
\E{\mo{\Phi}(\param)\vo{\Phi}^T(\param)} \\= \diag\begin{pmatrix}\E{\phi^2_0(\param)} &\cdots& \E{\phi^2_N(\param)}\end{pmatrix} =: \W.\eqnlabel{gp:W}
\end{multline}
Let us denote the optimal coefficients matrix determined using the GP formulation by $\F_{\text{GP}}$ which follows from \eqn{gp:errProj}
\begin{align}
\F_\text{GP} = \E{\vo{f}(\param)\vo{\Phi}^T(\param)}\W^{-1}. \eqnlabel{pcAnalytical}
\end{align}
The first and second order moments of $\fhatp$ are computed as follows,
\begin{align}
\E{\fhatp} & = \F\E{\mo{\Phi}(\param)} = \F\begin{bmatrix}1 & 0 & \cdots & 0\end{bmatrix}^T,  \eqnlabel{gp:meanFhat} 
\end{align}
\begin{align}
\E{\fhatp\fhatpt} & = \F\E{\mo{\Phi}(\param)\mo{\Phi}^T(\param)}\F^T = \F\W\F^T. \eqnlabel{gp:covFhat}
\end{align}
Higher order moments can be similarly computed from $\fhatp$, and will be prone to approximation errors.



We next present the result for determining PC coefficients that exactly match first and second order statistics of $\fp$.

\begin{theorem}\label{thm:gp1}
The PC coefficients that provide exact first and second order statistics are given by the columns of
\begin{align}
\F = \begin{bmatrix} \E{\fp} &  \vo{L}\vo{U}\W_1^{-1/2} \end{bmatrix}, \eqnlabel{gp:twoOrderExact}
\end{align}
where
\begin{align*}
\W_1 &= \mathbf{diag}\begin{pmatrix}\E{\phi^2_1(\param)} &\cdots& \E{\phi^2_N(\param)}\end{pmatrix},\\
\vo{L}\vo{L}^T &= \E{\fp\fpt} -\E{\fp}\E{\fpt}, 
\end{align*}
$\vo{U}\in\real^{n\times N}$ is an arbitrary matrix satisfying $\vo{U}\vo{U}^T = \I{n}$, and $N\ge n$.
\end{theorem}
\begin{proof}
The exactness of the mean estimated using GP formulation can be verified by examining \eqn{pcAnalytical}.
Since $\phi_0(\param)=1$, it readily follows that the first column of $\F_\text{GP}$ is $\E{\fp}$. 

The error in the second moment can be set to zero, if the following constraint is satisfied
\begin{align}
\F\W\F^T = \E{\fp\fpt}. \eqnlabel{gp:quadConstr}
\end{align}
Let us partition,
\begin{align}
\F := \begin{bmatrix} \E{\fp} & \F_1 \end{bmatrix}, \eqnlabel{gp:partF}
\end{align}
where $\F_1\in\real^{n\times N}$. This directly satisfies the mean constraint. Partitioning $\W$ as
$$
\W := \begin{bmatrix} \E{\phi_0(\param)} & \vo{0}\\\vo{0} & \W_1 \end{bmatrix} =  \begin{bmatrix} 1 & \vo{0}\\\vo{0} & \W_1 \end{bmatrix},
$$
\eqn{gp:quadConstr} becomes
\begin{align}
\F_1\W_1\F_1^T = \E{\fp\fpt} -\E{\fp}\E{\fpt}. \eqnlabel{gp:quadConstr1}
\end{align}
Since $\E{\fp\fpt} -\E{\fp}\E{\fpt} \ge 0$, we employ Cholesky factorization to write
\begin{align*}
\vo{L}\vo{L}^T &= \E{\fp\fpt} -\E{\fp}\E{\fpt}.
\end{align*}
Equation \Eqn{gp:quadConstr1} is a standard linear algebra problem with solution
\begin{align}
\F_1 = \vo{L}\vo{U}\W_1^{-1/2},\eqnlabel{gp:soln}
\end{align}
and $\vo{U}\in\real^{n\times N}$ is an arbitrary (rectangular) unitary matrix, i.e., $\vo{U}\vo{U}^T = \I{n}$.  Therefore, the PC coefficients that provide first and second moments with no error are given by the columns of

\begin{align*}
\F = \begin{bmatrix} \E{\fp} &  \vo{L}\vo{U}\W_1^{-1/2} \end{bmatrix}.
\end{align*}
\end{proof}

\vspace{-10pt}
The solution of $\F$ in \eqn{gp:twoOrderExact} is quite different from the solution in \eqn{pcAnalytical} that is obtained via GP approach.
The variable $\vo{U}$ parameterizes a family of solutions for $\F$ that exactly recovers the first and second moments. Clearly, for rows of $\vo{U}$ to be orthonormal, we require $N\ge n$.

If only first and second order statistics are required, then we can choose $N=n$ and $\vo{U}=\I{n}$ resulting in
\begin{align*}
\F = \begin{bmatrix} \E{\fp} &  \vo{L}\W_1^{-1/2} \end{bmatrix}. 
\end{align*}
For $N>n$, there are extra degrees of freedom that can be used to minimize other errors.

In the standard $\mathcal{L}_2$ formulation, the coefficients $\F$ are determined such that $\E{\|\fp-\fhatp\|_2}$ is minimized, and results in \eqn{gp:errProj}. The number of equations in \eqn{gp:errProj} is equal to the number of unknowns, which results in a unique solution for $\F$.
Adding constraints for first and second order statistics will result in more constraints than variables, and thus \eqn{gp:errProj} cannot be exactly satisfied. Therefore, optimal coefficients can be obtained via a constrained minimization of the residual error, i.e.,
\begin{align}
\min_{\vo{U}} \E{\vo{e}^T(\param)\vo{e}(\param)}, \text{ subject to } \vo{U}\vo{U}^T=\I{n}. \eqnlabel{QCopt}
\end{align}
Note that $\E{\vo{e}^T(\param)\vo{e}(\param)} = \tr{\E{\vo{e}(\param)\vo{e}^T(\param)}},$ and
\begin{align*}
&\E{\vo{e}(\param)\vo{e}^T(\param)} \\ & = \E{\Big(\fp-\F\mo{\Phi}^T(\param)\Big)\Big(\fp-\F\mo{\Phi}^T(\param)\Big)^T},\\
& = \vo{Q} - \F_1\R^T - \R\F_1^T + \F_1\W_1\F_1^T,
\end{align*}
where
\begin{align*}
\vo{Q} & := \E{(\fp-\E{\fp})(\fp-\E{\fp})^T},
\end{align*}
\begin{align}
\R & := \E{\fp\vo{\Phi}^T_{1}(\param)}, \eqnlabel{R}
\end{align}

$\F_1$ depends on $\vo{U}$ as given by \eqn{gp:soln}, and $\vo{\Phi}_{1}(\param)$ is the sub-vector of $\vo{\Phi}(\param)$ without the first element, i.e.,
$$
\vo{\Phi}_{1}(\param) := \begin{bmatrix}\phi_1(\param) \\ \vdots \\ \phi_N(\param)\end{bmatrix}.
$$
Therefore, the optimization problem in \eqn{QCopt} can be written as
\begin{multline}
\min_{\vo{U}\in\real^{n\times N}} \trace{\vo{Q} - \F_1\R^T - \R\F_1^T + \F_1\W_1\F_1^T}, \\ \text{ subject to } \vo{U}\vo{U}^T=\I{n}. \eqnlabel{nonConvex1}
\end{multline}
The optimization problem \eqn{nonConvex1} is non convex due to the constraint $\vo{U}\vo{U}^T=\I{n}$. However, it is a quadratically constrained quadratic programming problem, which can be converted to a convex optimization problem using various relaxations techniques \cite{bao2011semidefinite, anstreicher2012convex, park2017general} and solved with existing solvers \cite{cvxpy,cvxpy_rewriting}.

We propose an alternative approach to solve \eqn{nonConvex1}: solution \eqn{pcAnalytical} is first computed and subsequently projected on to the constraint set $\vo{U}\vo{U}^T=\I{n}$. 
The new formulation is formally presented as the following theorem.

\begin{theorem}
\label{thm:bestL2}
The coefficients that result in the $\mathcal{L}_2$-optimal PC approximation, subject to constraints on first and second-order statistics, are given by
\begin{align}
\F := \begin{bmatrix} \E{\fp} &  \vo{L}\vo{U}^\ast\W_1^{-1/2} \end{bmatrix},
\end{align}
where $\vo{U}^\ast :=  \vo{M}_1\vo{T}\vo{M}^T_2$, $\vo{M}_1$ and $\vo{M}_2$ are unitary matrices obtained from the singular value decomposition of $\vo{L}^{-1}\R\W_1^{-1/2}$, i.e.,
$$
\vo{L}^{-1}\R\W_1^{-1/2} = \vo{M}_1\vo{D}\vo{M}^T_2,
$$
$\vo{T}:=\begin{bmatrix}\I{n} & \vo{0}_{n\times(N-n)}\end{bmatrix}$, and $N\ge n$.
\end{theorem}

\begin{proof}
Let $\F_\text{GP}$ denote the optimal solution \eqn{pcAnalytical} for the standard GP approach.
The corresponding partitioned $\F_1$ is $\F_{1\text{GP}}:=\R\W_1^{-1}$ and the corresponding objective value is
$$
J_\text{GP} :=  \tr{\E{\vo{e}_\text{GP}(\param)\vo{e}^T_\text{GP}(\param)}} = \trace{\Q - \R\W_1^{-1}\R^T}.
$$
The objective function $J_{\vo{U}}$ for an arbitrary $\vo{U}\in\real^{n\times N}$ and the corresponding $\F$ in \eqn{gp:twoOrderExact} is given by
\begin{align*}
J_{\vo{U}} &:=  \tr{\E{\vo{e}_{\vo{U}}(\param)\vo{e}_{\vo{U}}(\param)}},\\
& = \trace{\Q + \vo{L}\vo{U}\vo{U}^T\vo{L}^T - \R\W_1^{-1/2}\vo{U}^T\vo{L}^T}  \\ &-\trace{\vo{L}\vo{U}\W_1^{-1/2}\R^T}.
\end{align*}
The difference in costs $J_{\vo{U}}$ and $J_\text{GP}$ is
\begin{align}
&\notag J_{\vo{U}} - J_\text{GP} =  \textbf{tr}\left(\vo{L}\vo{U}\vo{U}^T\vo{L}^T - \R\W_1^{-1/2}\vo{U}^T\vo{L}^T \right. \\ 
& \quad \left. - \vo{L}\vo{U}\W_1^{-1/2}\R^T + \R\W_1^{-1}\R^T\right),\\
& = \trace{\left(\vo{L}\vo{U}-\R\W_1^{-1/2}\right)\left(\vo{L}\vo{U}-\R\W_1^{-1/2}\right)^T}.\eqnlabel{gp:U>GP}
\end{align}
Therefore, $J_{\vo{U}} - J_\text{GP} = 0$ is satisfied for $\vo{U} =  \vo{U}_\text{GP}$, where
\begin{align*}
\vo{U}_\text{GP}&:=\vo{L}^{-1}\R\W_1^{-1/2}.
\end{align*}
However, in general $\vo{U}_\text{GP}\vo{U}^{T}_\text{GP} \neq \I{n}$.
Therefore, we project $\vo{U}_\text{GP}$ on the constraint set $\vo{U}\vo{U}^T = \I{n}$ as follows.
We decompose $\vo{U}_\text{GP}$ using singular-value decomposition: $\vo{U}_\text{GP} = \vo{M}_1\vo{D}\vo{M}^T_2$, where $\vo{D}=\begin{bmatrix}\vo{\Lambda} & \vo{0}_{n\times(N-n)}\end{bmatrix}$, and $\vo{\Lambda}$ is diagonal matrix with singular values of $\vo{U}_\text{GP}$. 
An optimal $\vo{U}^\ast$ subject to $\vo{U}^\ast\vo{U}^{\ast^T} = \I{n}$, is recovered as $\vo{U}^\ast := \vo{M}_1\vo{T}\vo{M}^T_2$, where  $\vo{T}:=\begin{bmatrix}\I{n} & \vo{0}_{n\times(N-n)}\end{bmatrix}$.
\end{proof}
Note, \eqn{gp:U>GP} implies that $J_{\vo{U}} \geq J_\text{GP}$ for any $\vo{U}\neq\vo{U}_\text{GP}$~. Therefore, enforcing the second-moment constraint increases the approximation error in the $\mathcal{L}_2$ sense.

\Fig{conGPC} shows the accuracy of the proposed constrained $\mathcal{L}_2$ method. We compare the moment errors for the standard GP formulation with the proposed formulation for increasing approximation order for a few candidate functions. The errors are plotted on a semi-log scale, lower bounded by the machine precision. We use Legendre polynomials as basis functions since $\Delta$ is uniformly distributed over $[-1,1]$.

For the polynomial case $f(\Delta) := \Delta^8$, the first-moment errors are negligible (bounded by machine precision) under both the standard GP and the constrained $\mathcal{L}_2$ formulations. However, the second-moment errors are substantial for the standard GP method and only vanish once the approximation order reaches eight. 
For non-polynomial functions, including $f(\Delta):=\tfrac{1}{1+\Delta+\Delta^2}$, $f(\Delta):=\sin^2(3\Delta)$, and $f(\Delta):=e^{-10\Delta^2}$, the standard GP formulation exhibits large second-moment errors even at high approximation orders. 
In contrast, the constrained $\mathcal{L}_2$ formulation yields second-moment errors at machine precision across all cases, even with low-order expansions. 
Thus, the constrained $\mathcal{L}_2$ approach enables the construction of lower-order surrogate models that capture the first two statistical moments exactly.
\Fig{conGPC} also shows the errors in $3^\text{rd}$ and $4^\text{th}$ moments. The errors for the standard GP and the constrained $\mathcal{L}_2$ formulations are nearly identical. Thus, the proposed formulation has little impact on the accuracy of higher-order moments relative to the standard GP method.

\section{Constrained $l_2$-optimal Approximation}\label{sec:conl2}

\begin{figure*}[tb]
    \centering
    \includegraphics[width=0.9\textwidth]{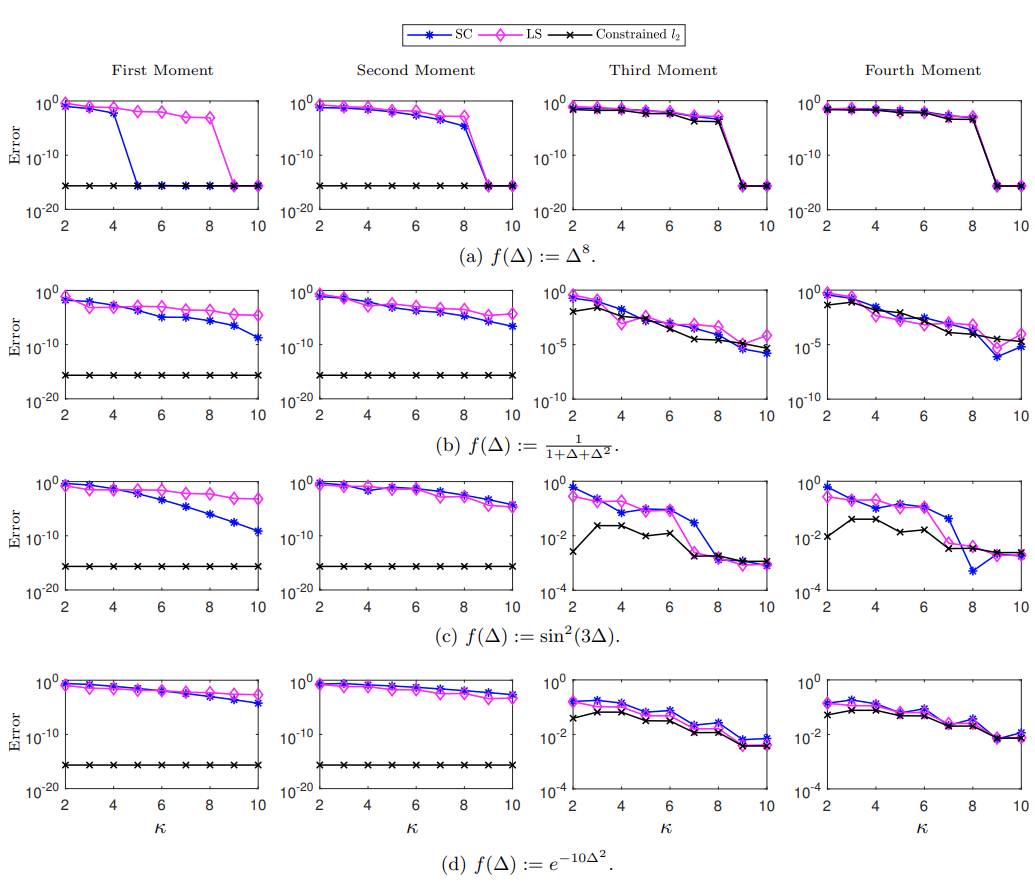}
    \caption{{Accuracy of PC approximations of different orders ($\kappa$) for various functions using standard SC, LS, and constrained $l_2$ formulations. Plots show absolute errors $\left|\mathbb{E}\big[f^m(\param)\big] - \mathbb{E}\big[\hat{f}^m(\param)\big]\right|$, for $m=1$, $2$, $3$, and $4$ respectively, where $\hat{f}$ is the approximated function, and $\Delta$ is uniformly distributed over $[-1,1]$.}}
    \label{fig:conSC}
\end{figure*}

Stochastic collocation is a non-intrusive method for uncertainty quantification that approximates the dependence of a system’s output on random inputs by evaluating the equations at finite number of sampled points.
The basis functions $\{\psi_i(\param)\}$ are selected to be multivariate polynomials interpolating over the data points $(\param_j,\vo{f}_j)$ defined by
\begin{align}
\vo{f}_j:=\vo{f}(\param_j), \eqnlabel{sc:fi}
\end{align}
where $\{\param_j\}_{j=0}^{N_{\text{SC}}}$ are discrete points in $\domainD$. Typically, the interpolation is achieved using
$$
\fp \approx \vo{\hat{f}}_{\text{SC}}(\param) := \sum_{i=0}^{N_{\text{SC}}} \vo{f}_i \psi_i(\param) =: \F_{\text{SC}}\vo{\Psi}(\param),
$$
where $\vo{\Psi}(\param)$ is a vector of Lagrange polynomials, $\psi_i(\param)$, defined by
\begin{multline*}
\psi_i(\param) := \prod_{j=0,j\neq i}^{N_{\text{SC}}} \frac{(\Delta_1-\Delta_{1j})\cdots(\Delta_d-\Delta_{dj})}{(\Delta_{1i}-\Delta_{1j})\cdots(\Delta_{di}-\Delta_{dj})}, \\ i=0,1,\cdots,N_{\text{SC}}.
\end{multline*}
Note, $\param_j := [\Delta_{1j}, \Delta_{2j}, \cdots, \Delta_{dj}]^T \in \Real^d$ and $\F_{\text{SC}} := [\vo{f}_0, \vo{f}_1, \cdots, \vo{f}_{N_{\text{SC}}}]$.

In SC, the moments of $\fp$ are approximated by computing the moments of $\vo{\hat{f}}_{\text{SC}}(\param)$, but these need not coincide with the true moments. Since $(N_{\text{SC}}+1)$ collocation points yield exactly $(N_{\text{SC}}+1)$ interpolation polynomials $\psi_i(\Delta)$ of degree $N_{\text{SC}}$, the formulation admits no additional degrees of freedom for enforcing moment constraints. Therefore, we turn to the least-squares (LS) method—also referred to as linear regression or point collocation \cite{walters2003towards,hosder2007efficient}—which we extend below to incorporate moment constraints.


First, we present the standard LS solution. We begin with selecting a suitable orthogonal polynomial basis $\vo{\Phi}(\param)$ such as defined in \eqn{Phi:def}. Then, $\fhat(\param) = \F\mo{\Phi}(\param)$.
A grid $G_{l_{2}}:=\{\param_i\}_{i=1}^{N_p}\in\domainD$ is created by generating $N_p$ samples from the random space $\domainD$. The LS solution for $\F$ is determined by the least squares fit, 
\begin{align*}
\min_{\F}\sum_{i=1}^{N_p}\|\F\vo{\Phi}(\param_i) - \vo{f}(\param_i)\|^2.
\end{align*}
The optimal $\F$ for this optimization problem is denoted as $\F_{\text{LS}}$, and is given by $\F_{\text{LS}} := \H_1^T\H_2^{-1}$,
where $\H_1 :=  \sum_{i=1}^{N_p}\vo{\Phi}(\param_i)\vo{f}(\param_i)^T$, and 
$\H_2 := \sum_{i=1}^{N_p}\vo{\Phi}(\param_i)\vo{\Phi}^T(\param_i)$.

Note, the existence of such $\F_{\text{LS}}$ requires that $N_p \geq (N+1)$. For the purpose of numerical experiments presented in this paper, we have used $N_p = 2(N+1)$, as recommended in \cite{hosder2007efficient}.

We next present an $l_2$-optimal approximation of $\fp$ with moment constraints. The moment constraints on $\F$ are imposed as (see \eqn{gp:meanFhat}, \eqn{gp:covFhat}, \eqn{gp:quadConstr} and \eqn{gp:partF})
\begin{align}
\F\begin{bmatrix}1 & 0 & \cdots & 0\end{bmatrix}^T = \E{\fp},  \eqnlabel{sc:mean} \\
 \F\W\F^T =  \E{\fp\fp^T}. \eqnlabel{sc:cov}
\end{align}
In $l_2$ formulation, the reference moments in \eqn{sc:mean} and \eqn{sc:cov} are computed by evaluating $\fp$ over a finite grid $G:=\{\param_j\}\in\domainD$, e.g., $\E{\fp} =  \sum_j w_j \vo{f}(\param_j)$ where $w_j$ is the weight associated with the grid point $\param_j$.

The $l_2$-optimal approximation of $\fp$ over the grid $G_{l_{2}}:=\{\param_i\}_{i=1}^{N_p}\in\domainD$ that recovers the first and second order statistics of $\fp$ is obtained from the following optimization problem,
\begin{equation} \eqnlabel{LS_optim_prob}
\begin{aligned}
\min_{\F}\sum_{i=1}^{N_p}\|\F\vo{\Phi}(\param_i) - \vo{f}(\param_i)\|^2, \\
\text{ subject to \eqn{sc:mean} and \eqn{sc:cov}. }
\end{aligned}
\end{equation}
The grid $G_{l_{2}}$ need not coincide with the grid $G$ used for computing true moments and may, in fact, be sparser. The resulting optimization is nonconvex due to the constraint in \eqn{sc:cov}. Following the $\mathcal{L}_2$ formulation, we reformulate the problem with unitary matrices, solve it without the nonconvex constraint, and project the solution onto the feasible set. The $l_2$-optimal solution with first- and second-moment constraints is stated in the theorem below.

\begin{theorem}
The coefficients that solve \eqn{LS_optim_prob} and result in the $l_2$-optimal PC approximation subject to constraints on first and second-order statistics, are given by
\begin{align}
\F := \begin{bmatrix} \E{\fp} &  \vo{L}\vo{U}^\ast\W_1^{-1/2} \end{bmatrix},
\end{align}
where $\vo{U}^\ast :=  \vo{M}_1\vo{T}\vo{M}^T_2$, $\vo{M}_1$ and $\vo{M}_2$ are  unitary matrices obtained from the singular value decomposition of $\hat{\vo{U}}$ defined in \eqn{sc:U1}, $\vo{T}:=\begin{bmatrix}\I{n} & \vo{0}_{n\times(N-n)}\end{bmatrix}$, and $N_p \geq (N+1)$. 
\label{thm:bestl2}
\end{theorem}
\begin{proof}
We define $\vo{e}(\param) := \vo{f}(\param) - \F\vo{\Phi}(\param)$, and seek an optimal $\F$ that minimizes 
$\trace{\sum_{i=1}^{N_p} \vo{e} (\param_i)\vo{e}^T(\param_i)}$, subject to \eqn{sc:mean} and \eqn{sc:cov}.
Using the results of Theorem \ref{thm:gp1}, constraint \eqn{sc:mean}  is satisfied when $\F$ is partitioned as $\F = \begin{bmatrix}\E{\fp} & \F_1\end{bmatrix}$, and \eqn{sc:cov} is satisfied when $\F_1 = \L\vo{U}\W_1^{-1/2}$, for any $\vo{U}$ satisfying $\vo{U}\vo{U}^T=\I{n}$.
Now, let us write the cost function as a function of $\F_1$
\begin{align*}
    J(\F_1) := \trace{\sum_{i=1}^{N_p} \vo{e}(\param_i)\vo{e}^T(\param_i)} 
\end{align*}
\begin{align*} 
     J(\F_1) = \trace{\EE_0 - \EE_1^T\F_1^T - \F_1\EE_1 + \F_1\EE_2\F_1^T},
\end{align*}
where
\begin{align*}
\EE_0 & := \sum_{i=1}^{N_p}\big(\vo{f}(\param_i) - \E{\fp}\big)\big(\vo{f}(\param_i) - \E{\fp}\big)^T, \\
\EE_1 & := \sum_{i=1}^{N_p}\vo{\Phi}_{1}(\param_i)\big(\vo{f}(\param_i) - \E{\fp}\big)^T, \\
\EE_2 &:= \sum_{i=1}^{N_p}\vo{\Phi}_{1}(\param_i)\vo{\Phi}^T_{1}(\param_i).
\end{align*}
Therefore, $J(\F_1)$ is minimized when
$$
\frac{\partial}{\partial \F_1} \trace{\EE_0 - \EE_1^T\F_1^T - \F_1\EE_1 + \F_1\EE_2\F_1^T} = 0,
$$
or
\begin{equation}
\F_1 = \EE_1^T\EE_2^{-1}. \eqnlabel{sc:F1}
\end{equation}
Equating the solution from \eqn{sc:F1} with the $\F_1$ that satisfies \eqn{sc:cov}, we get
\begin{align}
& \L\vo{U}\W_1^{-1/2} = \EE_1^T\EE_2^{-1} \\
\text{or } \quad & \hat{\vo{U}} := \L^{-1}\EE_1^T\EE_2^{-1}\W_1^{1/2}. \eqnlabel{sc:U1}
\end{align}
However, $\hat{\vo{U}}$ may not necessarily satisfy the constraint $\vo{U}\vo{U}^T = \I{n}$ that is necessary for matching the second moment.
Therefore, similar to the projection approach that we use in Theorem \ref{thm:bestL2}, let the singular-value decomposition of $\hat{\vo{U}}$ be $\vo{M}_1\vo{D}\vo{M}^T_2$, i.e.,
\begin{align}
\vo{M}_1\vo{D}\vo{M}_2^T = \L^{-1}\EE_1^T\EE_2^{-1}\W_1^{1/2},
\eqnlabel{sc:svd}
\end{align}
where $\vo{D}=\begin{bmatrix}\vo{\Lambda} & \vo{0}_{n\times(N-n)}\end{bmatrix}$, and $\vo{\Lambda}$ is diagonal matrix with singular values of $\hat{\vo{U}}$. Defining $\vo{U}^\ast := \vo{M}_1\vo{T}\vo{M}^T_2$  with $\vo{T}:=\begin{bmatrix}\I{n} & \vo{0}_{n\times(N-n)}\end{bmatrix}$,
ensures $\vo{U}^\ast\vo{U}^{\ast T} = \I{n}$. 
\end{proof}
\fig{conSC} compares moment errors from the proposed constrained $l_2$ formulation with those of SC and LS methods. As expected, the constrained $l_2$ approach yields zero error in the first moment across all approximation orders, while SC and LS exhibit significant errors at low orders. For the polynomial case in subplot (a), these errors vanish only at the $5^{\text{th}}$ and $9^{\text{th}}$ orders for SC and LS, respectively.

Second-moment errors from SC and LS are comparable and remain large for non-polynomial functions, whereas the constrained $l_2$ method achieves machine-precision accuracy even at low orders. Thus, like the constrained $\mathcal{L}_2$ formulation, it enables construction of low-order surrogates that exactly capture the first two moments. Errors in higher moments remain comparable to SC and LS, consistent with the imposed constraints.

\section{Conclusions} \label{sec:conclusions}
Standard Galerkin projection (GP) and least-squares (LS) methods often require high-order polynomial chaos (PC) expansions to achieve accurate moment estimates. In this work, we introduced constrained variants of these methods, namely, the constrained $\mathcal{L}_2$ and constrained $l_2$ formulations, which enforce exact matching of the first two moments. 
These formulations enable construction of lower-order PC surrogates that remain statistically accurate while being computationally more efficient. Given the widespread use of PC expansions for uncertainty propagation and quantification in control, the constrained approaches provide a promising pathway for generating low-order surrogate models. Application of these formulations to control problems will be pursued in future work.

\bibliographystyle{unsrt} 
\bibliography{references}
\end{document}